\documentclass[conference]{IEEEtran}
\IEEEoverridecommandlockouts
\usepackage{amsmath,amsfonts}
\usepackage{algorithmic}
\usepackage{algorithm}
\usepackage{array}
\usepackage{textcomp}
\usepackage{stfloats}
\usepackage{url}
\usepackage{verbatim}
\usepackage{graphicx}
\usepackage{cite}
\usepackage{amsthm}
\usepackage{subcaption}

\usepackage{amssymb}
\usepackage{xcolor}

\newtheorem{definition}{Definition}
\newtheorem{theorem}{Theorem}
\newtheorem{lemma}[theorem]{Lemma}         

\usepackage[margin=0.72in]{geometry}
\def\BibTeX{{\rm B\kern-.05em{\sc i\kern-.025em b}\kern-.08em
    T\kern-.1667em\lower.7ex\hbox{E}\kern-.125emX}}
\begin{document}

\title{SpecFed: Accelerating Federated LLM Inference with Speculative Decoding and Compressed Transmission\\
\vspace{-2mm}
\thanks{ 
This work is supported by the National Science and Technology Major Project of China under Grant No. 2025ZD1304900,
and in part by the National Key R\&D Program of China (Grant No.2024YFE0200801, No.2024YFE0200804). The authors would also like to thank Dr.~Ke Zhang from Waseda University for his valuable assistance with the implementation of the code.
}
}

\author{\IEEEauthorblockN{
Ce Zheng\IEEEauthorrefmark{1}, Xinghan Wang\IEEEauthorrefmark{1}, Jiahong Ning\IEEEauthorrefmark{2}, Yuxuan Shi\IEEEauthorrefmark{1}, Ning Huang\IEEEauthorrefmark{1}, Tingting Yang\IEEEauthorrefmark{1}
}

\IEEEauthorblockA{\IEEEauthorrefmark{1}Department of Broadband Communication, Pengcheng Laboratory, Shenzhen 518055, China}
\IEEEauthorblockA{\IEEEauthorrefmark{2}  Dalian Maritime University, Dalian 116026, China}
\IEEEauthorblockA{Email: \{zhengc, wangxh03\}@pcl.ac.cn, jiahong.ning@mnsu.edu,\\ \{shiyx01, huangn01, yangtt\}@pcl.ac.cn}
}

\maketitle

\begin{abstract}
Federated inference enhances LLM performance in edge computing through weighted averaging of distributed model predictions. However, autoregressive LLM inference requires frequent full-model forward passes across workers, severely limiting decoding throughput. Distributed deployment further aggravates this due to a communication bottleneck: each worker must transmit full token probability distributions per draft token, dominating end-to-end latency. To address these challenges, we introduce speculative decoding to enable parallel LLM processing and propose a top-$K$ compressed transmission scheme with two server-side reconstruction strategies. We theoretically analyze the robustness of our method in terms of local reconstruction error, aggregation bias, and acceptance-rate bias, and derive corresponding bounds. Experiments demonstrate that our scheme achieves high generation fidelity while significantly reducing communication overhead.
\end{abstract}


\section{Introduction}
Federated inference — also known as ensemble inference — improves accuracy and robustness by aggregating predictions from models hosted on different devices via weighted averaging~\cite{shlezinger2021collaborative, malka2024decentralized, zhou2025towards, yilmaz2022over, kumazawa2024toward}. This paradigm has gained significant traction in distributed AI systems, particularly for edge computing scenarios where models reside on heterogeneous devices with varying computational capabilities. 

In the context of large language models (LLMs), federated inference frameworks have also been explored to leverage the collective power of distributed models for enhanced generation quality~\cite{yu2024breaking, chen2025harnessing, yun2025ensemble, yaodetermine}.
However, the autoregressive nature of LLM inference introduces a critical challenge in federated settings: each token generation step requires all workers to execute a full forward pass, leading to frequent synchronization and significantly reduced decoding throughput~\cite{fu2025fast}. To mitigate the inefficiency of autoregressive token generation, we consider \textit{speculative decoding}~\cite{leviathan2023fast,chen2023accelerating}, where a small “draft” model autoregressively generates a sequence of $\gamma$ candidate tokens, which are then sent to the distributed LLMs for parallel decoding and verification. By shifting part of the sequential workload to parallel processing, speculative decoding reduces the number of full LLM forward passes, thereby alleviating the latency inherent in purely autoregressive decoding while preserving generation quality.

Despite its effectiveness in reducing LLM invocation frequency, speculative decoding introduces a new challenge in federated settings: \textbf{a significant communication bottleneck}. In each round, every worker must transmit its full-token probability distribution for all draft tokens to the server for aggregation and verification. Given large vocabulary sizes (e.g., 32K or more) and high-precision representations (e.g., FP16), this results in substantial uplink traffic per worker—on the order of hundreds of kilobits per token. The total communication cost scales linearly with both the number of workers and the draft length, which can easily dominate end-to-end latency and undermine the throughput gains from parallel decoding~\cite{zhao2024edge, zhu2025efficient}.
Recent work has explored communication-efficient techniques for distributed speculative decoding, including quantization~\cite{ning2025dssd} and split verification~\cite{zhang2025quantize}. However, these approaches do not specifically address the unique challenges of federated LLM inference, where multiple heterogeneous models must collaboratively generate tokens through a weighted aggregation mechanism that requires precise distribution alignment.

In this paper, we study \textbf{federated LLM inference with speculative decoding}. Building on prior top-$K$ compression techniques~\cite{oh2025uncertainty,zheng2025communication,zheng2026fast}, we propose a communication-efficient transmission scheme that sends only the top-$K$ token probabilities per worker, along with two server-side reconstruction strategies. Our approach preserves generation fidelity while drastically cutting bandwidth usage. 
We theoretically analyze the robustness of our method in terms of local reconstruction error, aggregation bias, and acceptance-rate bias, and derive corresponding bounds, which are further validated through experiments.

The remainder of this paper is organized as follows: Section~\ref{sec:sys_model} introduces our system model. Section~\ref{sec:compressed_transmission} details our compressed transmission scheme with top-$K$ truncation and reconstruction. 
Section~\ref{sec:theory} presents our theoretical analysis of token distribution distortion and acceptance rate variation.
Finally, Section~\ref{sec:conclusion} concludes the paper.

\section{System Model}
\label{sec:sys_model}
\begin{figure*}
\centering
\includegraphics[width=0.8\linewidth]{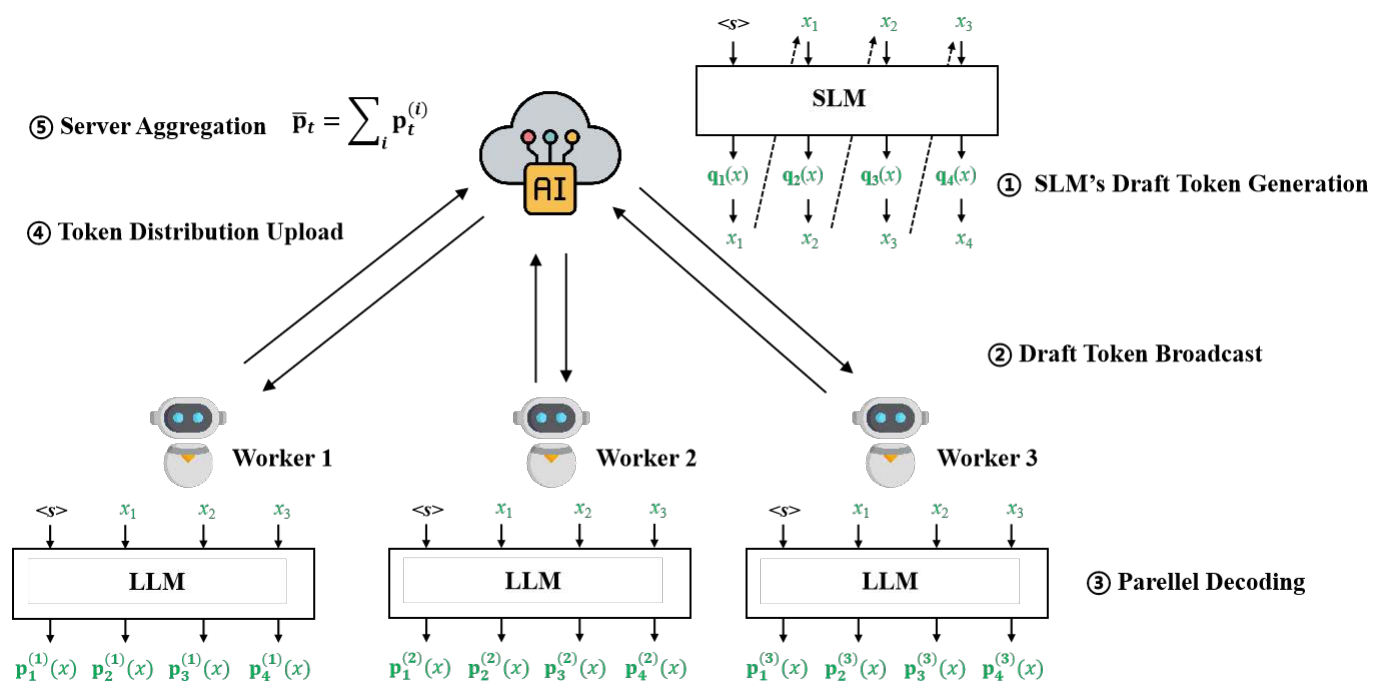}
\caption{An illustration of federated LLM inference with Speculative Decoding.}
 \label{fig:FedSpec}
\end{figure*}
\subsection{Federated LLM Inference with Speculative Decoding}
\label{sec:FI}
We consider a \textbf{federated LLM inference} framework consisting of a coordinating server and a set of distributed workers 
\(\mathcal{M} = \{1,2,\ldots,M\}\), where each worker hosts a local LLM. All LLMs share the same vocabulary $\mathcal{V}$, though the model sizes and architectures may be heterogeneous. 
The inference task follows an \textbf{autoregressive decoding} paradigm. At each step, the server sends the current prefix or newly generated token to all workers. Each worker performs a forward pass to compute a probability distribution over the vocabulary, and the server aggregates the distributions to generate the next token. Due to this sequential dependency, every generated token requires all workers to execute a forward pass and upload a full token distribution, causing frequent synchronization and significant communication overhead, especially for large vocabularies, which limits decoding throughput.

To mitigate this problem, we further deploy a small language model (SLM) at the server—also using the same vocabulary—and adopt the \textbf{speculative decoding} mechanism as illustrated in Fig.~\ref{fig:FedSpec}, which proceeds as follows.

\noindent \textbf{Step 1: SLM's Draft Token Generation.}
Given the prefix, the SLM autoregressively generates a draft token sequence of length $\gamma$, denoted by \( \mathbf{x} = [x_1, x_2, \dots, x_\gamma]\).
At each step $t = 1, \dots, \gamma$, the draft token is sampled as
\[
x_t \sim \mathbf{q}_t, \quad t = 1, \cdots, \gamma,
\] 
where \(\quad \mathbf{q}_t = \big[q_t(x)\big]_{x \in \mathcal{V}} \) denotes the output distribution of the SLM at step $t$.

\noindent \textbf{Step 2: Draft Token Broadcast.}
The draft sequence $\mathbf{x}$ is broadcast to all workers $i \in \mathcal{M}$

\noindent \textbf{Step 3: Parellel Decoding.}
Given the draft sequence, each worker computes the conditional probability distribution in parallel
\[
\mathbf{p}_{t,i} = \big[p_{t,i}(x)\big]_{x \in \mathcal{V}} \in \mathbb{R}^{|\mathcal{V}|}, \quad t=1,\cdots,\gamma+1,
\]
where $p_{t,i}(x)$ denotes the probability of generating token $x$ at step $t$ by worker $i$.

\noindent \textbf{Step 4: Token Distribution Upload.} 
Each worker uploads its output token distribution~\(\mathbf{p}_{t,i}\) to the server. 

\noindent \textbf{Step 5: Server Aggregation} 
The server aggregates the received distribution via a weighted average:
\begin{equation}
    \bar{\mathbf{p}}_t = \sum_{i \in \mathcal{M}} w_{i} \mathbf{p}_{t,i}
    = \big[\bar{p}_{t}(x)\big]_{x \in \mathcal{V}},
\end{equation}
where \(\bar{p}_t(x) = \sum_{i \in \mathcal{M}} w_{i} \, p_{t,i}(x).\)

\noindent \textbf{Step 6: Draft token Verification.}
The server verifies the draft token in two sub-steps: 
\textbf{\textit{a).~Accept/Reject:}} the draft token is accepted if $q_t(x_t)<\bar{p}_t(x_t)$, and otherwise rejected with probability ${\bar{p}_t(x_t)}/{q_t(x_t)}$. Hence, we have the expected acceptance rate as~\cite{leviathan2023fast}:
\begin{equation}
\label{eq:acceptance_rate}
\alpha_t \!=\! \mathbb{E}_{x\sim \mathbf{q}_t} \!\!\left[ \min\left( \frac{\bar{p}_t(x)}{{q_t(x)}}, 1\right) \!\right] \!=\! \sum_{x \in \mathcal{V}} \min\left( \bar{p}_t(x), {q_t(x)}\right).
\end{equation}

\noindent \textbf{\textit{b).~Resample:}}
If $x_t$ is rejected, a new token $x'_t$ is resampled from the \emph{residual distribution}
\begin{equation}
\label{eq:resample}
    r_t(x) = \frac{\max(0, \bar{p}_t(x)-q_t(x))}{\sum_{v \in \mathcal{V}} \max(0, \bar{p}_t(v)-q_t(v))}.
\end{equation}

If all draft tokens are accepted, the server samples an additional token $x'_{\gamma+1}\sim\bar{\mathbf{p}}_{\gamma+1}$; otherwise, at the first rejected step $t\le\gamma$, it samples $x'_t$ from the residual distribution in~\eqref{eq:resample} and discards $\{x_{t+1},\dots,x_\gamma\}$. The verified sequence $[x_1,\dots,x'_t]$, with $t=\gamma+1$ in the full-acceptance case, is appended to the prefix for the next decoding iteration.

\subsection{Communication Bottleneck}
While speculative decoding can reduce the number of full LLM forward passes required at each worker, each worker still needs to upload a probability distribution over all tokens for every draft token to the server. This incurs significant communication overhead, which grows with both the number of workers and the vocabulary size. To alleviate this issue, it is necessary to compress the transmitted distributions without severely degrading reconstruction quality.

\section{Compressed Transmission Scheme}
\label{sec:compressed_transmission}
To address the communication bottleneck described above, we adopt \textbf{top-$K$ truncation} for probability transmission. Specifically, at decoding step $t$, worker $i$ produces a probability distribution $\mathbf{p}_{t,i} = [p_{t,i}(x)]_{x \in \mathcal{V}}$. Let $\mathcal{V}^{(k_i)} \subset \mathcal{V}$ denote the set of top-$k_i$ tokens with the highest probabilities under $\mathbf{p}_{t,i}$.

\noindent\textbf{Compressed Transmission Scheme:}
Worker $i$ transmits $\mathcal{V}^{(k_i)}$ along with $\{p_{t,i}(x)\}_{x \in \mathcal{V}^{(k_i)}}$,
i.e., the identities of the top-$k_i$ tokens and their associated probabilities.

With a slight abuse of notation, we denote the reconstructed distribution under different reconstruction schemes by
$p^{(k_i)}_{t,i}(x)$.
To simplify notation, we further define
\begin{equation}
\label{eq:rho_epsilon}
\rho_{t,i} = \!\!\sum_{v \in \mathcal{V}^{(k_i)}} p_{t,i}(v), 
\qquad
\epsilon_{t,i} = \!\!\sum_{v \notin \mathcal{V}^{(k_i)}} p_{t,i}(v) = 1 - \rho_{t,i},
\end{equation}
where $\rho_{t,i}$ denotes the retained probability mass and $\epsilon_{t,i}$ is the residual mass after top-$k_i$ truncation.

As top-$K$ truncation removes some probability mass, the server cannot fully determine the original distribution. 
We therefore consider two reconstruction schemes to approximate the original distribution.

\subsection{Renormalized Top-$K$ Reconstruction}
In the first scheme, the server renormalizes the received probabilities over $\mathcal{V}^{(k_i)}$, while assigning zero probability to all remaining tokens.
The reconstructed distribution is given by
\begin{equation}
\label{eq:renormalize}
p^{(k_i)}_{t,i}(x) =
\begin{cases}
\dfrac{p_{t,i}(x)}{\rho_{t,i}}, & x \in \mathcal{V}^{(k_i)}, \\[8pt]
0, & x \notin \mathcal{V}^{(k_i)} .
\end{cases}
\end{equation}
This scheme preserves the relative probabilities among the top-$k_i$ tokens, but discards all probability mass outside $\mathcal{V}^{(k_i)}$.

\subsection{Residual-Mass Redistribution}
In the second scheme, the server preserves the original probabilities of the transmitted top-$k_i$ tokens and redistributes the remaining probability mass uniformly over the rest of the vocabulary.
The reconstructed distribution is defined as
\begin{equation}
\label{eq:residual_mass}
p^{(k_i)}_{t,i}(x) =
\begin{cases}
p_{t,i}(x), & x \in \mathcal{V}^{(k_i)}, \\[6pt]
\dfrac{\epsilon_{t,i}}{|\mathcal{V}| - k_i}, & x \notin \mathcal{V}^{(k_i)} .
\end{cases}
\end{equation}

This scheme preserves the original probabilities of the top-$k_i$ tokens while maintaining a valid probability distribution over the entire vocabulary.

After reconstructing the distributions from all workers, the server aggregates them to form a global token distribution.
Specifically, let $\mathbf{k} = [k_1,\ldots,k_M]$ denote the top-$K$ profile across workers.
At decoding step $t$, the aggregated distribution is computed as
\begin{equation}
\bar{\mathbf{p}}^{(\mathbf{k})}_t = \sum_{i \in \mathcal{M}} w_i \, \mathbf{p}^{(k_i)}_{t,i},
\end{equation}
where $w_i \ge 0$ and $\sum_{i \in \mathcal{M}} w_i = 1$.
The resulting distribution $\bar{\mathbf{p}}^{(\mathbf{k})}_t$ is then used for token sampling or draft-token verification.

\section{Theoretical Analysis}
\label{sec:theory}
In this section, we analyze the robustness of the proposed \emph{compressed transmission scheme} in federated speculative decoding.
Specifically, we characterize how top-$K$ compression and reconstruction affect (i) the distortion of local and aggregated token distributions, and (ii) the variation of the acceptance rate in speculative decoding.\footnote{Extending the analysis to end-to-end communication cost, decoding throughput, and final generation quality is left for future work. These metrics are fundamentally coupled with the token distribution and acceptance rate, and our results already provide indicative insights into their behavior.}

To this end, we first define error measures at both the local and aggregated levels, as well as the acceptance rate variation, and then derive corresponding bounds.

\begin{definition}[Local Reconstruction Error]
The local reconstruction error for worker $i$ at decoding step $t$ is
\begin{equation}
\Delta^{(k_i)}_{t,i} \triangleq \big\| \mathbf{p}^{(k_i)}_{t,i} - \mathbf{p}_{t,i} \big\|_1.
\end{equation}
\end{definition}

\begin{definition}[Aggregation Bias]
The aggregation bias induced by top-$K$ reconstruction is
\begin{equation}
\label{eq:BBias}
\Delta^{(\mathbf{k})}_t \triangleq  \big\| \bar{\mathbf{p}}^{(\mathbf{k})}_t - \bar{\mathbf{p}}_t \big\|_1.
\end{equation}
\end{definition}
\begin{definition}[Acceptance Rate Variation]
The acceptance rate variation at decoding step $t$ is defined as
\begin{equation}
\Delta \alpha_t^{(\mathbf{k})} \triangleq 
\left| \alpha_t^{(\mathbf{k})} - \alpha_t \right|.
\end{equation}
\end{definition}

\subsection{Local Reconstruction Error}
We first characterize  the error introduced by reconstructing each worker's top-$K$ compressed distribution and quantify the resulting $\ell_1$ reconstruction error in the following lemma.
\begin{lemma}[Local Top-$K$ Reconstruction Error Bound]
\label{lem:local_error}
For worker $i$ at decoding step $t$, the $\ell_1$ reconstruction error $\Delta^{(k_i)}_{t,i}$
satisfies
\begin{equation}
\Delta^{(k_i)}_{t,i} = 2\epsilon_{t,i}^{(k_i)}
\end{equation}
under the reconstructed distribution defined in
\eqref{eq:renormalize}, and
\begin{equation}
\Delta^{(k_i)}_{t,i} \le 2\epsilon_{t,i}^{(k_i)}
\end{equation}
under the reconstructed distribution defined in \eqref{eq:residual_mass}.
\end{lemma}

\begin{proof}
By definition,
\begin{equation}
\Delta^{(k_i)}_{t,i}
= \sum_{x \in \mathcal{V}}
\big| p^{(k_i)}_{t,i}(x) - p_{t,i}(x) \big|.
\end{equation}
Splitting the vocabulary into $\mathcal{V}^{(k_i)}$ and its complement,
we obtain
\begin{align}
\Delta^{(k_i)}_{t,i}
=& \!\!\!\sum_{x \in \mathcal{V}^{(k_i)}}  \!\!\!\left|p^{(k_i)}_{t,i}(x) - p_{t,i}(x) \right| \!\!+ \!\!\!\! \sum_{x \notin \mathcal{V}^{(k_i)}} \!\!\!\left|p^{(k_i)}_{t,i}(x) - p_{t,i}(x) \right|.
\end{align}

\textit{\textbf{1). Renormalized Top-$K$ Reconstruction.}}
For this scheme, we have
\begin{align}
\Delta^{(k_i)}_{t,i}
\overset{(a)}{=}& \sum_{x \in \mathcal{V}^{(k_i)}} \left|\dfrac{p_{t,i}(x)}{\rho_{t,i}} - p_{t,i}(x) \right|
+ \sum_{x \notin \mathcal{V}^{(k_i)}} p_{t,i}(x) \notag\\
=& \frac{1-\rho_{t,i}}{\rho_{t,i}}
\sum_{x \in \mathcal{V}^{(k_i)}} p_{t,i}(x)
+ \sum_{x \notin \mathcal{V}^{(k_i)}} p_{t,i}(x) \notag\\
\overset{(b)}{=}& 2\epsilon_{t,i}^{(k_i)}.
\end{align}
where $(a)$ follows from \eqref{eq:renormalize} and $(b)$ follows from \eqref{eq:rho_epsilon}.

\textit{\textbf{2). Residual-Mass Redistribution.}}
For this scheme, we have
\begin{align}
\Delta^{(k_i)}_{t,i}
\overset{(c)}{=}& \sum_{x \notin \mathcal{V}^{(k_i)}} \left| \dfrac{\epsilon_{t,i}}{|\mathcal{V}| - k_i} - p_{t,i}(x) \right| \notag\\
\leq& \sum_{x \notin \mathcal{V}^{(k_i)}}
\left(
\frac{\epsilon_{t,i}}{|\mathcal{V}|-k_i}
+ p_{t,i}(x)
\right)
= 2\epsilon_{t,i}^{(k_i)}.
\end{align}
where $(c)$ follows from \eqref{eq:residual_mass}.

This completes the proof.
\end{proof}
As shown in Lemma~\ref{lem:local_error}, the local reconstruction error introduced by top-$K$ compression is bounded, indicating that the distortion of individual worker distributions is controlled.
\subsection{Aggregation Bias}
We now analyze the overall bias introduced when aggregating the reconstructed distributions from all workers. 
The following theorem provides an upper bound on the $\ell_1$ aggregation bias.
\begin{theorem}[Aggregation Bias Upper Bound]
\label{thm:aggregation_bias}
The aggregation bias $\Delta^{(\mathbf{k})}_t$ induced by top-$K$ compression is upper bounded as
\begin{equation}
\Delta^{(\mathbf{k})}_t
\leq
2 \sum_{i \in \mathcal{M}} w_i \epsilon_{t,i}^{(k_i)}.
\end{equation}
\end{theorem}

\begin{proof}
Using the triangle inequality, we have
\begin{align}
\Delta^{(\mathbf{k})}_t
= \left\|
\sum_{i \in \mathcal{M}} w_i
\big( \mathbf{p}^{(k_i)}_{t,i} - \mathbf{p}_{t,i} \big)
\right\|_1 
\leq \sum_{i \in \mathcal{M}} w_i
\big\| \mathbf{p}^{(k_i)}_{t,i} - \mathbf{p}_{t,i} \big\|_1.
\end{align}
Applying Lemma~\ref{lem:local_error} completes the proof.
\end{proof}
Together with the local reconstruction bound, these results guarantee that top-$K$ truncation introduces a bounded distortion both at the individual worker level and at the aggregated distribution level.

\subsection{Acceptance Rate}
The acceptance rate plays a crucial role: a higher acceptance rate implies that more tokens are accepted, leading to fewer interactions between the SLM and LLMs, and consequently reducing the amount of data transmitted. 
\begin{theorem} [Bound on the acceptance rate variation]
\label{the:delta_bound}
Let $\mathbf{\alpha}_t^{(\mathbf{k})}$ denote the acceptance rate via top-$k$ compressed transmission. $\Delta \alpha$ is bounded by
\begin{equation}
\label{eq:acceptance_bound}
    \Delta \alpha_t^{(\mathbf{k})} \leq \frac{1}{2}\Delta^{(\mathbf{k})}_t \leq \sum_{i \in \mathcal{M}} w_i \epsilon_{t,i}^{(k_i)}
\end{equation}
\end{theorem}
\begin{proof}
According to  Lemma~3.3 in~\cite{leviathan2023fast}, we have
$TV(\mathbf{q}_t,\bar{\mathbf{p}}_t)=1-\sum_{x\in\mathcal{V}} \min(q_t(x),\bar{p}_t(x))$, and $TV\left(\mathbf{q}_t,\bar{\mathbf{p}}_t^{(\mathbf{k})}\right)=1-\sum_{x\in\mathcal{V}} \min\left(q_t(x),\bar{p}^{(\mathbf{k})}_t(x)\right)$, where $TV(\cdot,\cdot)$ denotes the total variation distance.
Combining the above results with \eqref{eq:acceptance_rate}, we obtain
\begin{align}
    \Delta \alpha_t^{(\mathbf{k})} &= \left|  \sum_{x\in\mathcal{V}} \min \!\left(q_t(x),\bar{p}^{(\mathbf{k})}_t(x)\right)- \!\!\sum_{x\in\mathcal{V}} \min(q_t(x),\bar{p}_t(x)) \right|\notag\\
    &= \left|TV\left(\mathbf{q}_t,\bar{\mathbf{p}}_t^{(\mathbf{k})}\right) - TV(\mathbf{q}_t,\bar{\mathbf{p}}_t) \right| \notag\\
    &\overset{(a)}{\leq} TV \left(\bar{\mathbf{p}}_t^{(\mathbf{k})},\bar{\mathbf{p}}_t \right)=\frac{1}{2}\Delta^{(\mathbf{k})}_t \overset{(b)}{\leq} \sum_{i \in \mathcal{M}} w_i \epsilon_{t,i}^{(k_i)}.
\end{align}
Here $(a)$ follows from the triangle inequality of the total variation distance (Proposition~1 in~\cite{markatou2018non}), and $(b)$ follows from Theorem~\ref{the:delta_bound}
\end{proof}
Theorem~\ref{the:delta_bound} shows that the perturbation of the acceptance behavior caused by compressed transmission is also bounded, ensuring stable decoding efficiency.

\section{Experiment}
\label{sec:experiment}
In this section, we validate the compressed transsmission scheme and assess the effects of top-$K$ truncation on federated LLM inference.
Existing evaluations, which compare final task accuracy or output similarity under identical prompts suffer from two limitations:
First, token sampling introduces randomness, and the autoregressive nature of LLMs causes such randomness to propagate across generation steps, resulting in output variability even for identical inputs. 
Second, prior evaluations based on final task accuracy do not capture lossy inference bias at the token or distribution level, since generated outputs may be correct but not aligned with the LLM's output~\cite{bachmann2025judge}. Accordingly, we redesign the experiments to explicitly evaluate distribution-level distortions. 

\subsection{Experiment Setup}
We first perform federated LLM inference without compressed transmission to obtain the full token sequence together with the corresponding LLM vocabulary distributions at each generation step. Then we obtain the reconstructed distribution from \eqref{eq:renormalize} and \eqref{eq:residual_mass}. The bias is then computed according to \eqref{eq:BBias} and averaged across all steps and samples to quantify the overall deviation. 

Our federated LLM inference scenario involves two workers hosting heterogeneous LLMs: \texttt{LLaMA-7B} and \texttt{LLaMA-13B}, respectively. The coordinating server deploys the SLM, \texttt{LLaMA-68M}, for speculative decoding. To emulate the server–worker architecture, the SLM and LLMs are executed on NVIDIA A800 GPUs. All experiments are conducted on the \texttt{wmt14\_ende\_de} dataset~\cite{bojar-EtAl:2014:W14-33}, which consists of instruction-following examples widely used for LLM evaluation.
We further set the weights $w_1 = w_2 = 1/2$ and $k_1 = k_2$ for the ease of analysis.

\subsection{Simulation results}
Temperature controls the sharpness of the output vocabulary distribution by rescaling logits before softmax, redistributing probability mass while preserving token ranking~\cite{renze2024effect, li2025exploring}. Since vocabulary truncation operates on ranked token probabilities, temperature directly affects the residual mass beyond the top-$k$ tokens; thus, we vary the temperature to evaluate performance under different distribution concentration levels.

We compute the average aggregation bias under the two reconstruction strategies as
\begin{equation}
    \bar{\Delta} = \mathbb{E}[\Delta^{(\mathbf{k})}_t].
\end{equation}
Similarly, for the residual mass and acceptance rate bias, we have
\begin{equation}
    \bar{\epsilon} = \mathbb{E}\left[ \sum_{i \in \mathcal{M}} w_i \epsilon_{t,i}^{(k_i)} \right].
\end{equation}
\begin{equation}
    \Delta\bar{\alpha} = \mathbb{E}\left[ \Delta \alpha_t^{(\mathbf{k})} \right].
\end{equation}

\begin{figure*}[hbtp]
    \centering
    \begin{subfigure}{0.33\textwidth}
        \centering
        \includegraphics[width=\linewidth]{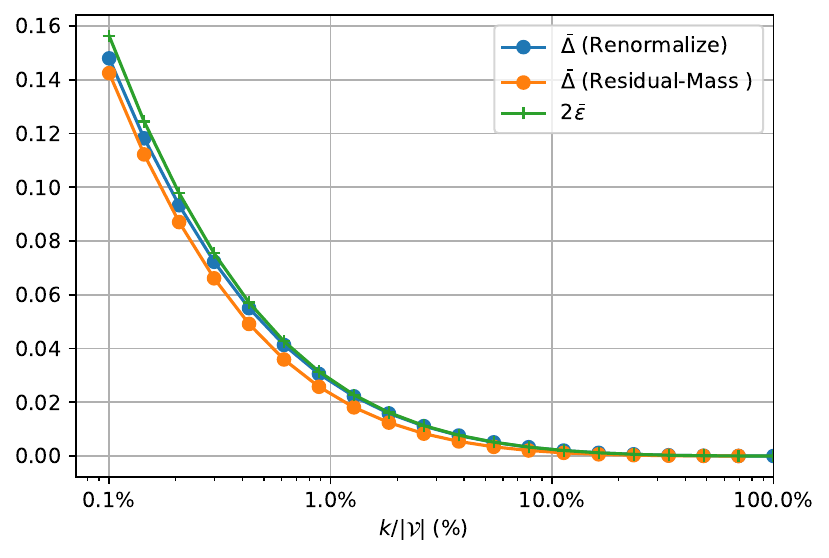}
        \caption{$T=0.8$}
        \label{fig:case1}
    \end{subfigure}\hfill
    \begin{subfigure}{0.33\textwidth}
        \centering
        \includegraphics[width=\linewidth]{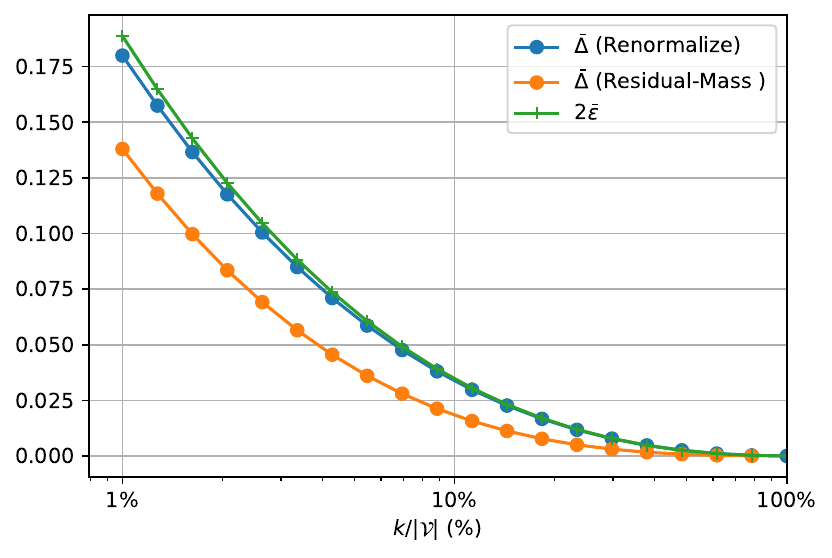}
        \caption{$T=1$}
        \label{fig:case2}
    \end{subfigure}\hfill
    \begin{subfigure}{0.33\textwidth}
        \centering
        \includegraphics[width=\linewidth]{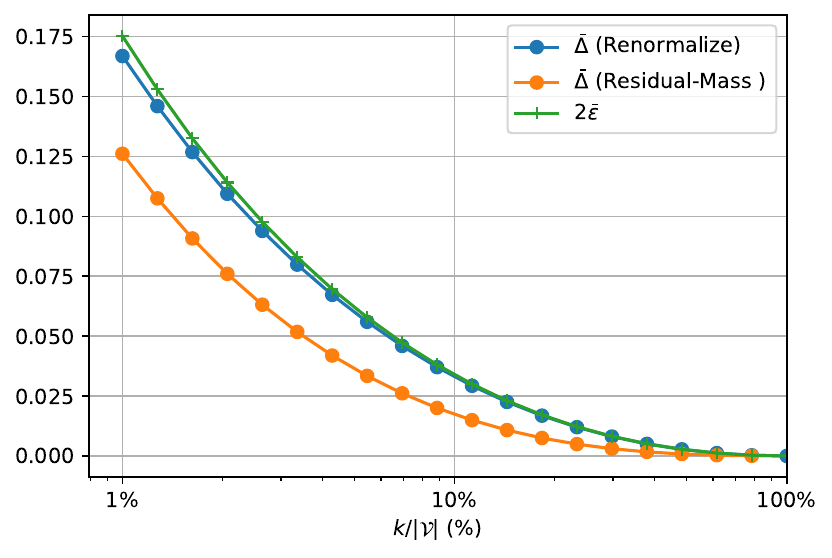}
        \caption{$T=1.2$}
        \label{fig:case3}
    \end{subfigure}
    \caption{Average aggregation bias $\bar{\Delta}$ under varying communication payloads ($K / |\mathcal{V}|$ in \%, with $|\mathcal{V}|=32,000$).}
    \label{fig:three_cases}
\end{figure*}

\begin{figure*}[hbtp]
    \centering
    \begin{subfigure}{0.33\textwidth}
        \centering
        \includegraphics[width=\linewidth]{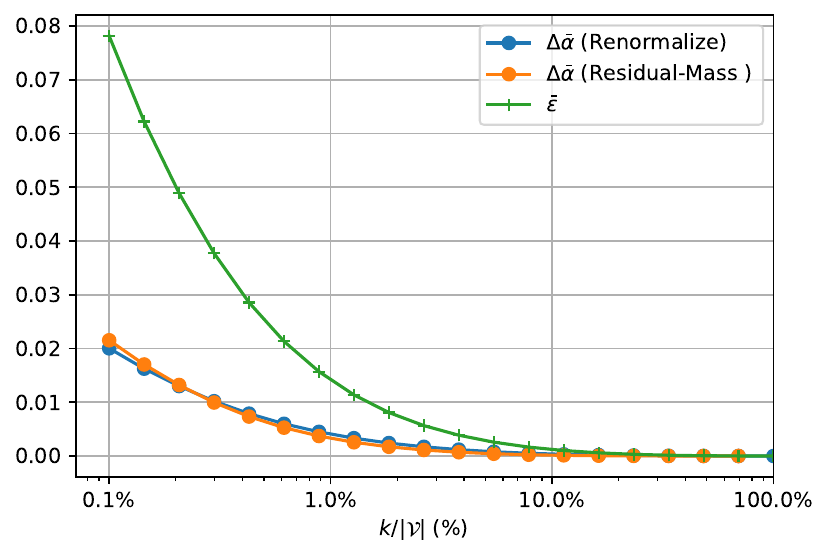}
        \caption{$T=0.8$}
        \label{fig:case1}
    \end{subfigure}\hfill
    \begin{subfigure}{0.33\textwidth}
        \centering
        \includegraphics[width=\linewidth]{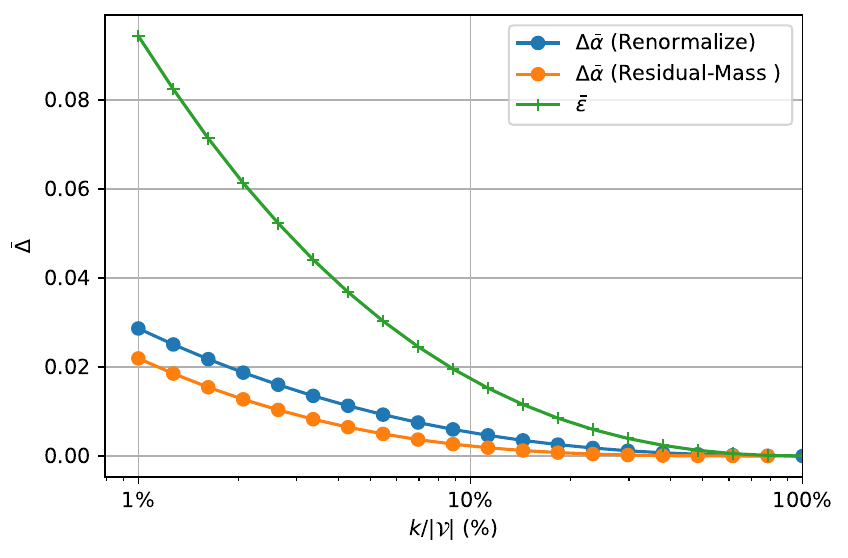}
        \caption{$T=1$}
        \label{fig:case2}
    \end{subfigure}\hfill
    \begin{subfigure}{0.33\textwidth}
        \centering
        \includegraphics[width=\linewidth]{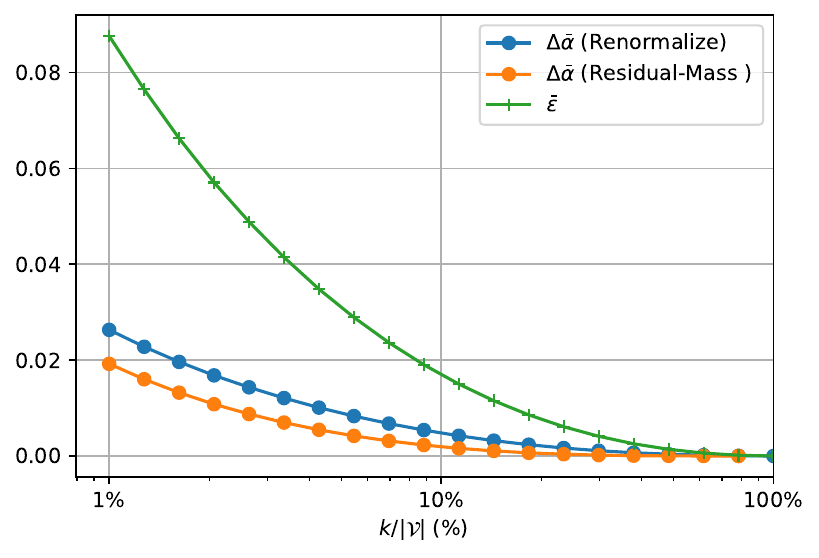}
        \caption{$T=1.2$}
        \label{fig:case3}
    \end{subfigure}
    \caption{Average acceptance rate bias $\Delta\bar{\alpha}$ under varying communication payloads ($K / |\mathcal{V}|$ in \%, with $|\mathcal{V}|=32,000$).}
    \label{fig:three_cases_acceptance}
\end{figure*}

The results are shown in Figs.~\ref{fig:three_cases} and~\ref{fig:three_cases_acceptance}. In Fig.~\ref{fig:three_cases}, we observe that for both reconstruction strategies, the average aggregation bias $\bar{\Delta}$ is bounded by $2\bar{\epsilon}$, confirming the validity of Theorem~\ref{thm:aggregation_bias}. Moreover, the \textit{Residual-Mass Redistribution} strategy consistently exhibits a smaller bias than the \textit{Renormalized Top-$K$ Reconstruction} strategy, which aligns with the intuition from Lemma~\ref{lem:local_error}: $\Delta^{(k_i)}{t,i} = 2 \epsilon{t,i}^{(k_i)}$ for the \textit{Residual-Mass Redistribution} strategy, whereas $\Delta^{(k_i)}{t,i} \le 2 \epsilon{t,i}^{(k_i)}$ for the \textit{Renormalized Top-$K$ Reconstruction} strategy.

In Fig.~\ref{fig:three_cases_acceptance}, we further see that $\Delta\bar{\alpha}$ is also bounded by $\bar{\epsilon}$. Note that this bound is not tight, and $\bar{\epsilon}$ under the \textit{Residual-Mass Redistribution} strategy is not necessarily smaller than that under the \textit{Renormalized Top-$K$ Reconstruction} strategy (e.g., for $K=320$, corresponding to only $0.1\%$ of the vocabulary). Nevertheless, these observations confirm the practicality of compressed transmission.

Furthermore, across different temperatures $T$, even under aggressive compression (e.g., $K=320$, i.e., only $1\%$ of the vocabulary), both $\bar{\Delta}$ and $\Delta\bar{\alpha}$ remain small, demonstrating that the proposed scheme achieves favorable communication--quality trade-offs.

\section{Conclusion}
\label{sec:conclusion}
In this paper, we studied federated LLM inference with speculative decoding, addressing the communication bottleneck inherent in distributed autoregressive generation. By leveraging a top-$K$ compressed transmission scheme combined with two server-side reconstruction strategies, our approach significantly reduces communication overhead while preserving generation fidelity. We theoretically analyzed the robustness of the method, deriving bounds on local reconstruction error, aggregation bias, and acceptance-rate bias, and validated these results through extensive experiments. The findings demonstrate that federated speculative decoding with compressed transmission is both effective and efficient, providing a practical solution for scalable edge deployment of LLMs.

\bibliographystyle{IEEEtran}
\bibliography{ref}

\vfill

\end{document}